\documentclass{amsart}
\usepackage{amssymb,amsmath,amsfonts}
\usepackage[foot]{amsaddr}
\usepackage{geometry}
\usepackage{multicol,multirow}
\usepackage{graphicx}
\usepackage{enumerate,enumitem,setspace}
\usepackage{verbatim}
\usepackage{pdfsync}
\usepackage{gastex}
\usepackage[lowtilde]{url}
\usepackage{xcolor}

\newtheorem{theorem}{Theorem}

\newtheorem{lemma}[theorem]{Lemma}
\newtheorem{proposition}[theorem]{Proposition}

\newtheorem{definition}[theorem]{Definition}
\theoremstyle{remark}
\newtheorem{remark}[theorem]{Remark}

\newcommand{\Wbf}{\mathbf{W}^{\ge 6}_{\mathrm{bf}}}

\begin{document}
\title{Complexity of regular bifix-free languages}
\author{Robert Ferens}
\email{robert.ferens@interia.pl}
\author{Marek Szyku{\l}a}
\email{msz@cs.uni.wroc.pl}
\address{Institute of Computer Science,\\
University of Wroc{\l}aw, Wroc{\l}aw, Poland}

\begin{abstract}
We study descriptive complexity properties of the class of regular bifix-free languages, which is the intersection of prefix-free and suffix-free regular languages.
We show that there exist a single ternary universal (stream of) bifix-free languages that meet all the bounds for the state complexity basic operations (Boolean operations, product, star, and reversal).
This is in contrast with suffix-free languages, where it is known that there does not exist such a stream.
Then we present a stream of bifix-free languages that is most complex in terms of all basic operations, syntactic complexity, and the number of atoms and their complexities, which requires a superexponential alphabet.

We also complete the previous results by characterizing state complexity of product, star, and reversal, and establishing tight upper bounds for atom complexities of bifix-free languages. We show that to meet the bound for reversal we require at least 3 letters and to meet the bound for atom complexities $n+1$ letters are sufficient and necessary. For the cases of product, star, and reversal we show that there are no gaps (magic numbers) in the interval of possible state complexities of the languages resulted from an operation; in particular, the state complexity of the product $L_m L_n$ is always $m+n-2$, while of the star is either $n-1$ or $n-2$.

\end{abstract}

\maketitle
\section{Introduction}

A language is \emph{prefix-free} or \emph{suffix-free} if no word in the language is a proper prefix or suffix, respectively, of another word from the language.
If a language is prefix-free and suffix-free then it is \emph{bifix-free}.
Languages with these properties have been studied extensively.
They form important classes of codes, whose applications can be found in such fields as cryptography, data compression, information transmission, and error correction methods.
In particular, \emph{prefix} and \emph{suffix codes} are prefix-free and suffix-free languages, respectively, while bifix-free languages can serve as both kinds of codes.
For a survey about codes see~\cite{BPR09,JuKo1997Codes}.
Moreover, they are special cases of \emph{convex languages} (see e.g.~\cite{BSX10} for the related algorithmic problems).
Here we are interested how the descriptive complexity properties of prefix-free and suffix-free languages are shared in their common subclass.

There are three natural measures of complexity of a regular language that are related to the Myhill (Myhill-Nerode) congruence on words.
The usual \emph{state complexity} or \emph{quotient complexity} is the number of states in a minimal DFA recognizing the language.
Therefore, state complexity measures how much memory we need to store the language in the form of a DFA, or how much time we need to perform an operation that depends on the size of the DFA.
Therefore, we are interested in finding upper bounds for complexities of the resulting languages obtained as a result of some operation (e.g. union, intersection, product, or reversal).
\emph{Syntactic complexity} measures the number of transformations in the transition semigroup or, equivalently, the number of classes of words that act differently on the states~\cite{BrYe11,Pin97}; this provides a natural bound on the time and space complexity of algorithms working on the transition semigroup (for example, a simple algorithm checking whether a language is \emph{star-free} just enumerates all transformations and verifies whether no one of them contains a non-trivial cycle \cite{McSe71}).
The third measure is called the \emph{complexity of atoms} \cite{BrTa14}, which is the number and state complexities of the languages of words that distinguish exactly the same subset of states (quotients).

Most complex languages and universal witnesses were proposed by Brzozowski in~\cite{Brz13}.
The point here is that, it is more suitable to have a single witness that is most complex in the given subclass of regular languages, instead of having separate witnesses meeting the upper bound for each particular measure and operation.
Besides theoretical aspects, this concept has also a practical motivation:
To test efficiency of various algorithms or systems operating on automata (e.g. computational package GAP~\cite{GAP4}), it is natural to use worst-case examples, that is, languages with maximal complexities.
Therefore, it is preferred to have just one universal most complex example than a set of separate example for every particular case.
Of course, it is also better to use a smallest possible alphabet.

It may be surprising that such a single witness exists for most of the natural subclasses of regular languages: the class of all regular languages~\cite{Brz13}, right-, left-, and two-sided ideals~\cite{BDL15}, and prefix-convex languages \cite{BrSi16PrefixConvex}.
However, there does not exist a single witness for the class of suffix-free languages \cite{BrSz15ComplexityOfSuffixFree}, where two different witnesses must be used.

In this paper we continue the previous studies concerning the class of bifix-free languages \cite{BrJiLiSm14QuotientComplexityOfBifixFactorSubwordFree,SzWi16SyntacticComplexityOfBifixFree}.
In~\cite{BrJiLiSm14QuotientComplexityOfBifixFactorSubwordFree} the tight bound on the state complexity of basic operations on bifix-free languages were established; however, the witnesses were different for particular cases.
The syntactic complexity complexity of bifix-free languages was first studied in~\cite{BLY12}, where a lower bound was established, and then the formula was shown to be an upper bound in~\cite{SzWi16SyntacticComplexityOfBifixFree}.

Our main contributions are as follows:
\begin{enumerate}
\item We show a single ternary witness of bifix-free languages that meets the upper bounds for all basic operations.
This is in contrast with the class of suffix-free languages, where such most complex languages do not exist.
\item We show that there exist most complex languages in terms of state complexity of all basic operations, syntactic complexity, and number of atoms and their complexities. It uses a superexponential alphabet, which cannot be reduced.
\item We prove a tight upper bound on the number of atoms and the quotient complexities of atoms of bifix-free languages.
\item We provide a complete characterization of state complexity for product and star, and show the exact ranges for the possible state complexities for product, star, and reversal of bifix-free languages.
\item We prove that at least a ternary alphabet must be used to meet the bound for reversal, and at an $(n+1)$-ary alphabet must be used to meet the bounds for atom complexities.
\end{enumerate}

\section{Preliminaries}

\subsection{Regular languages and complexities}

Let $\Sigma$ be a non-empty finite alphabet.
In this paper we deal with regular languages $L \subseteq \Sigma^*$.
For a word $w \in L$, the \emph{(left) quotient} of $L$ is the set $\{u \mid wu \in L\}$, which is also denoted by $L.w$.
Left quotients are related to the Myhill-Nerode congruence on words, where two words $u,v \in \Sigma^*$ are equivalent if for every $x \in \Sigma^*$, we have $ux \in L$ if and only if $vx \in L$.
Thus the number of quotients is the number of equivalence classes in this relation.
The number of quotients of $L$ is the \emph{quotient complexity} $\kappa(L)$ of this language~\cite{Brz10a}. A language is regular if it has a finite number of quotients.

Let $L,K \subseteq \Sigma^*$ be regular languages over the same alphabet $\Sigma$.
By \emph{Boolean operations} on these languages we mean \emph{union} $L \cup K$, \emph{intersection} $L \cap K$, \emph{difference} $L \setminus K$, and \emph{symmetric difference} $L \oplus K$.
The reverse language $L^R$ of $L$ is the language $\{a_k\ldots a_1 \mid a_1 \ldots a_k \in L, a_1,\ldots,a_k \in \Sigma\}$.
By the \emph{basic operations} on regular languages we mean the Boolean operations, the product (concatenation), the star, and the reversal operation.
By the \emph{complexity} of an operation we mean the maximum possible quotient complexity of the resulted language, given as a function of the quotient complexities of the operands.

The \emph{syntactic complexity} $\sigma(L)$ of $L$ is the number of equivalence classes of the Myhill equivalence relation on $\Sigma^+$, where two words $u,v \in \Sigma^+$ are equivalent if for any words $x,y \in \Sigma^*$, we have $xuy \in L$ if and only if $xvy \in L$.

The third measure of complexity of a regular language $L$ is the number and quotient complexities of $\emph{atoms}$~\cite{BrTa14}.
Atoms arise from the left congruence of words refined by Myhill equivalence relation: two words $u,v \in \Sigma^*$ are equivalent if for any word $x \in \Sigma^*$, we have $xu \in L$ if and only if $xv \in L$~\cite{Iva16ComplexityOfAtoms}.
Thus $u$ and $v$ are equivalent if they belong exactly to the same left quotients of $L$.
An equivalence class of this relation is an \emph{atom}~\cite{BrTa14} of $L$.
It is known that (see~\cite{BrTa14}) an atom is a non-empty intersection of quotients and their complements, and the quotients of a language are unions of its atoms.
Therefore, we can write $A_S$ for an atom, where $S$ is the set of quotients of $L$; then $A_S$ is the intersection of the quotients of $L$ from $S$ together with the complements of the quotients of $L$ outside $S$.

\subsection{Finite automata and transformations}

A \emph{deterministic finite automaton} (\emph{DFA}) is a tuple $\mathcal{D}=(Q,\Sigma,\delta,q_0,F)$, where $Q$ is a finite non-empty set of \emph{states}, $\Sigma$ is a finite non-empty \emph{alphabet}, $\delta\colon Q\times \Sigma\to Q$ is the \emph{transition function}, $q_0\in Q$ is the \emph{initial} state, and $F\subseteq Q$ is the set of \emph{final} states.
We extend $\delta$ to a function $\delta\colon Q\times \Sigma^*\to Q$ as usual: for $q \in Q$, $w \in \Sigma^*$, and $a \in \Sigma$, we have $\delta(q,\varepsilon) = q$ and $\delta(q,wa) = \delta(\delta(q,w),a)$, where $\varepsilon$ denotes the empty word.

A state $q \in Q$ is \emph{reachable} if there exists a word $w \in \Sigma^*$ such that $\delta(q_0,w) = q$. Two states $p,q \in Q$ are \emph{distinguishable} if there exists a word $w \in \Sigma^*$ such that either $\delta(p,w) \in F$ and $\delta(q,w) \notin F$ or $\delta(p,w) \notin F$ and $\delta(q,w) \in F$.

A DFA is \emph{minimal} if there is no DFA with a smaller number of states that recognizes the same language. It is well known that this is equivalent to that every state is reachable and every pair of distinct states is distinguishable.
Given a regular language $L$, all its minimal DFAs are isomorphic, and their number of states is equal to the number of left quotients $\kappa(L)$ (see e.g.~\cite{Brz10a}).
Hence, the quotient complexity $\kappa(L)$ is also called the \emph{state complexity} of $L$.
If a DFA is minimal then every state $q$ corresponds to a quotient of the language, which is the set words $w$ such that $\delta(q,w) \in F$. We denote this quotient by $K_q$.
We also write $A_S$, where $S$ is a subset of states, for
$$A_S = \bigcap_{q \in S} K_q \cap \bigcap_{q \in \overline{S}} \overline{K_q},$$
which is an atom if $A_S$ is non-empty.

A state $q$ is \emph{empty} if $K_q = \emptyset$.

Throughout the paper, by $\mathcal{D}_n$ we denote a DFA with $n$ states, and without loss of generality we always assume that its set of states $Q = \{0,\ldots,n-1\}$ and that the initial state is $0$.

In any DFA $\mathcal{D}_n$, every letter $a\in \Sigma$ induces a transformation $\delta_a$ on the set $Q$ of $n$ states.
By $\mathcal{T}_n$ we denote the set of all $n^n$ transformations of $Q$; then $\mathcal{T}_n$ is a monoid under composition.
For two transformations $t_1,t_2$ of $Q$, we denote its composition as $t_1 t_2$.
The transformation induced by a word $w \in \Sigma^*$ is denoted by $\delta_w$.
The \emph{image} of $q\in Q$ under a transformation $\delta_w$ is denoted by $q \delta_w$, and the \emph{image} of a subset $S \subseteq Q$ is $S \delta_w = \{q \delta_w \mid q \in S\}$.
The \emph{preimage} of a subset $S \subset Q$ under a transformation $\delta_w^{-1}$ is $S \delta_w^{-1} = \{q \in Q \mid q \delta_w \in S\}$.
Note that if $w = a_1 \ldots a_k$, then $\delta_{a_1 \ldots a_k}^{-1} = \delta_{a_k}^{-1} \ldots \delta_{a_1}^{-1}$.
The \emph{identity transformation} is denoted by $\mathbf{1}$, which is also the same as $\delta_\varepsilon$, and we have $q\mathbf{1} = q$ for all $q \in Q$.

The \emph{transition semigroup} $T(n)$ of $\mathcal{D}_n$ is the semigroup of all transformations generated by the transformations induced by $\Sigma$.
Since the transition semigroup of a minimal DFA of a language $L$ is isomorphic to the syntactic semigroup of $L$~\cite{Pin97}, syntactic complexity $\sigma(L)$ is equal to the cardinality $|T(n)|$ of the transition semigroup $T(n)$.

Since a transformation $t$ of $Q$ can be viewed as a directed graph with regular out-degree equal $1$ and possibly with loops, we transfer well known graph terminology to transformations:
The \emph{in-degree} of a state $q \in Q$ is the cardinality $|\{p \in Q \mid pt = q\}|$.
A \emph{cycle} in $t$ is a sequence of states $q_1,\ldots,q_k$ for $k \ge 2$ such that $q_i t = q_{i+1}$ for $i=1,\ldots,k-1$, and $q_k t = q_1$.
A \emph{fixed point} in $t$ is a state $q$ such that $q t = q$; we therefore do not call fixed points cycles.

A transformation that maps a subset $S$ to a state $q$ and fixes all the other states is denoted by $(S \to q)$. If $S$ is a singleton $\{p\}$ then we write shortly $(p \to q)$. A transformation that has a cycle $q_1,\ldots,q_k$ and fixes all the other states is denoted by $(q_1,\ldots,q_k)$.

A \emph{nondeterministic finite automaton} (\emph{NFA}) is a tuple $\mathcal{N} = (Q,\Sigma,\delta,I,F)$, where $Q$, $\Sigma$, and $F$ are defined as in a DFA, $I$ is the set of \emph{initial states}, and $\delta\colon Q \times \Sigma \cup \{\varepsilon\} \to 2^Q$ is the transition function.

\subsection{Most complex languages}

A \emph{stream} is a sequence $(L_k,L_{k+1},\dots)$ of regular languages in some class, where $n$ is the state complexity of $L_n$.
A \emph{dialect} $L'_n$ of a language $L_n$ is a language that differs only slightly from $L_n$. There are various types of dialects, depending what changes are allowed.
A \emph{permutational dialect} (or \emph{permutationally equivalent dialect}) is a language in which letters may be permuted or deleted.
Let $\pi\colon \Sigma \to \Sigma$ be a partial permutation.
If $L_n(a_1,\ldots,a_k)$ is a language over the alphabet $\Sigma=\{a_1,\ldots,a_k)$, then we write $L_n(\pi(a_1),\ldots,\pi(a_k))$ for a language in which a letter $a_i$ is replaced by $\pi(a_i)$. In the case a letter $a_i$ is removed, so not defined by $\pi(a_i)$, we write $\pi(a_i)=\_$.
For example, if $L=\{a,ab,abc\}$, then $L(b,a,\_)=\{b,ba\}$.

A stream is \emph{most complex} in its class if all their languages and all pairs of languages together with their dialects meet all the bounds for the state complexities of basic operations, the syntactic complexity, the number and the complexities of atoms.
Note that binary operations were defined for languages with the same alphabets.
Therefore, if the alphabet is not constant in the stream, to meet the bounds for binary Boolean operations, for every pair of languages we must use their dialects that restrict the alphabet to be the same.

Sometimes we restrict only to some of these measures.
In some cases, this allows to provide an essentially simpler stream over a smaller alphabet when we are interested only in those measures.
In particular, if a syntactic complexity requires a large alphabet and for basic operations it is enough to use a constant number of letters, it is desirable to provide a separate stream which is most complex just for basic operations.

Dialects are necessary for most complex streams of languages, since otherwise they would not be able to meet upper bounds in most classes. In particular, since $L_n \cup L_n = L_n$, the state complexity of union would be at most $n$ in this case.
Other kinds of dialects are possible (e.g.~\cite{BrSi16PrefixConvex}), though permutational dialects are the most restricted.

\subsection{Bifix-free languages}

A language $L$ is \emph{prefix-free} if there are no words $u,v \in \Sigma^+$ such that $uv \in L$ and $u \in L$. A language $L$ is \emph{suffix-free} if there are no words $u,v \in \Sigma^+$ such that $uv \in L$ and $v \in L$.
A language is \emph{bifix-free} if it is both prefix-free and suffix-free.

The following properties of minimal DFAs recognizing prefix-free, suffix-free, and bifix-free languages, adapted to our terminology, are well known (see e.g.~\cite{BrJiLiSm14QuotientComplexityOfBifixFactorSubwordFree,BLY12,CmJi12,SzWi16SyntacticComplexityOfBifixFree}):

\begin{lemma}\label{lem:bifix-free_properties}
Let $\mathcal{D}_n(Q,\Sigma,\delta,0,F)$ be a minimal DFA recognizing a  non-empty language $L$. Then $L$ is bifix-free if and only if:
\begin{enumerate}
\item[1] There is an empty state, which is $n-1$ by convention, that is, $n-1$ is not final and $(n-1)\delta_a = n-1$ for all $a \in \Sigma$.
\item[2] There exists exactly one final state, which is $n-2$ by convention, and its quotient is $\{\varepsilon\}$; thus $(n-2)\delta_a = n-1$ for all $a \in \Sigma$.
\item[3] For $u\in \Sigma^+$ and $q \in Q \setminus \{0\}$, if $q \delta_u \neq n-1$, then $0 \delta_u \neq q \delta_u$.
\end{enumerate}
The conditions~(1) and~(2) are sufficient and necessary for a prefix-free languages, and the conditions~(1) and~(3) are sufficient and necessary for a suffix-free language.
\end{lemma}

It follows that a minimal DFA recognizing a non-empty bifix-free language must have at least $n \ge 3$ states.

Since states $0$, $n-2$, and $n-1$ are special in the case DFAs of bifix-free languages, we denote the remaining ``middle'' states by $Q_M = \{1,\ldots,n-3\}$.
Condition~3 implies that suffix-free and so bifix-free are \emph{non-returning} (see~\cite{EHJ2016NonReturning}), that is, there is no non-empty word $w \in \Sigma^+$ such that $L.w = L$.

Note that in the case of unary languages, there is exactly one bifix-free language for every state complexity $n \ge 3$, which is $\{a^{n-2}\}$.
The classes of unary prefix-free, unary suffix-free, and unary bifix-free languages coincide and we refer to it as to \emph{unary free} languages.

The state complexity of basic operations on bifix-free languages was studied in~\cite{BrJiLiSm14QuotientComplexityOfBifixFactorSubwordFree}, where different witness languages were shown for particular operations.

The syntactic complexity of bifix-free languages was shown to be $(n-1)^{n-3}+(n-2)^{n-3}+(n-3)2^{n-3}$ for $n \ge 6$ \cite{SzWi16SyntacticComplexityOfBifixFree}.
Moreover, the transition semigroup of a minimal DFA $\mathcal{D}_n$ of a witness language meeting the bound must be $\Wbf(n)$, which is a transition semigroup containing three types of transformations and can be defined as follows:
\begin{definition}[The largest bifix-free semigroup]\label{def:Wbf}
\begin{eqnarray*}
\Wbf(n) & = & \{t \in T(n) \mid\\
\text{(type~1)} & & \{0,n-2,n-1\}t = \{n-1\}\text{ and }Q_M t \subset Q_M \cup \{n-2,n-1\}\text{, or}\\
\text{(type~2)} & & 0t = n-2\text{ and }\{n-2,n-1\}t=\{n-1\}\text{ and }Q_M t \subset Q_M \cup \{n-1\}\text{, or}\\
\text{(type~3)} & & 0t \in Q_M\text{ and }\{n-2,n-1\}t=\{n-1\}\text{ and }Q_M t \subseteq \{n-2,n-1\}\,\}.
\end{eqnarray*}
\end{definition}
Following \cite{SzWi16SyntacticComplexityOfBifixFree}, we say that an unordered pair $\{p, q\}$ of distinct states from $Q_M$ is \emph{colliding} in $T(n)$ if there is a transformation $t \in T(n)$ such that $0t = p$ and $rt = q$ for some $r \in Q_M$.
A pair of states is \emph{focused} by a transformation $u \in T(n)$ if $u$ maps both states of the pair to a single state $r \in Q_M \cup \{n-2\}$.
It is known that (\cite{SzWi16SyntacticComplexityOfBifixFree}) in semigroup $\Wbf(n)$ there are no colliding pairs and every possible pair of states is focused by some transformation, and it is the unique maximal transition semigroup of a minimal DFA of a bifix-free language with this property.


\section{Complexity of bifix-free languages}

In this section we summarize and complete known results concerning state complexity of bifix-free regular languages.

We start from the obvious upper bound for the maximal complexity of quotients.

\begin{proposition}\label{pro:quotients_upper_bound}
Let $L$ be a bifix-free language with state complexity $n$.
Each (left) quotients of $L$ have state complexity at most $n-1$, except $L$, $\{\varepsilon\}$, and $\emptyset$, which always have state complexities $n$, $2$, and $1$, respectively.
\end{proposition}
\begin{proof}
Since bifix-free languages are non-returning, their non-initial quotients have at most state complexity $n-1$.
\end{proof}

\subsection{Boolean operations}

In~\cite{BrJiLiSm14QuotientComplexityOfBifixFactorSubwordFree} it was shown that $mn-(m+n)$ (for $m,n \ge 4$) is a tight upper bound for the state complexity of union and symmetric difference of bifix-free languages, and that to meet this bound a ternary alphabet is required.
It was also shown there that $mn-3(m+n-4)$ and $mn-(2m+3n-9)$ (for $m,n \ge 4$) are tight upper bounds for intersection and difference, respectively, and that a binary alphabet is sufficient to meet these bounds.
Since the tight bound is smaller for unary free languages, the size of the alphabet cannot be reduced.

It may be interesting to observe that the alphabet must be essentially larger to meet the bounds in the case when $m=3$.
\begin{remark}
For $n \ge 3$, to meet the bound $mn-(m+n)$ for union or symmetric difference with minimal DFAs $\mathcal{D}'_3$ and $\mathcal{D}_n$ at least $n-2$ letters are required.
\end{remark}
\begin{proof}
For each $q \in \{1,\ldots,n-2\}$ state $(1',q)$ must be reachable in the product automaton. In $\mathcal{D}'_3$ state $0'$ is the only state that can be mapped to $1'$ by the transformation of some letter $a$. This means that at least $n-2$ different letters are required.
\end{proof}

\subsection{Product}

The tight bound for the product is $m+n-2$, which is met by unary free languages.
We show that there is no other possibly for the product of bifix-free languages, that is, $L_m L_n$ has always state complexity $m+n-2$.

\begin{theorem}\label{thm:product}
For $m \ge 3$, $n \ge 3$, for every bifix-free languages $L'_m$ and $L_n$, the product $L'_m L_n$ meets the bound $m+n-2$.
\end{theorem}
\begin{proof}
Let $\mathcal{D}_m' = (Q',\Sigma,\delta',0',\{(n-2)'\})$ and $\mathcal{D}_n = (Q,\Sigma,\delta,0,\{n-2\})$ be minimal DFAs for $L_m'$ and $L_n$, respectively. 
We use the well known construction for an NFA $\mathcal{N}$ recognizing the product of two regular languages.
Then $Q' \cup Q$ is the set of states, $0'$ is initial and $n-2$ is a final state.
We have $\delta^{\mathcal{N}}$ being the transition function such that $\delta^{\mathcal{N}}(p,a) = \{q\}$ whenever $\delta'(p,a)=q$ for $p,q \in Q'$, or $\delta(p,a)=q$ for $p,q \in Q$.
Also, we have the $\varepsilon$-transition $\delta^{\mathcal{N}}((m-2)',\varepsilon) = \{0\}$.
We determinize $\mathcal{N}$ to $\mathcal{D}_P$; since every reachable subset has exactly one state from $Q'$, we can assume that the set of states is $Q' \times 2^Q$, so $\mathcal{D}_P=(Q' \times 2^Q,\Sigma,\delta^P,\{0'\},Q' \times \{\{n-2\}\})$.

Since $m+n-2$ is an upper bound for product, it is enough to show that at least $m+n-2$ states are reachable and pairwise distinguishable in $\mathcal{D}_P$.
We show that the states $((n-2)',\{0\})$, $(p,\emptyset)$ for each $p \in \{0',\ldots,(m-3)'\}$, and $((m-1)',\{q\})$ for each $q \in \{1,\ldots,n-1\}$ are reachable and distinguishable.
Let $R$ be the set of these states.
From this place in the context of reachablitity and distinguishability we will consider only the states from $R$.

Since $L_{m'}$ is prefix-free and $\mathcal{D}_m'$ is minimal, every state $(p,\emptyset) \in R$ (where $p \in \{0',\ldots,(m-3)'\}$) is reached from $(0',\emptyset)$ by a word reaching state $p$ in $\mathcal{D}_m'$.
Furthermore, state $((m-2)',\{0\})$ is reached by a word $w'$ from non-empty language $L'_m$.
Every state $((m-1)',\{q\}) \in R$ (where $q \in \{1,\dots,n-1\}$) may be reached from state $((m-2)',\{0\})$ by a word $w$ such that $0\delta_w=q$, and hence by $w'w$ in $\mathcal{D}$.
It remains to show distinguishability.

Consider two distinct states $(r,\{q_1\})$ and $(r,\{q_2\})$ from $R$; then $r \in \{(m-2)',(m-1)'\}$ and $q_1,q_2 \in Q$.
These states are distinguishable by a word distinguishing $q_1$ and $q_2$ in $\mathcal{D}_n$.

Consider two states $(p,\emptyset)$ and $(r,\{q\})$ from $R$.
There exists a word $w$ such that $(p,\emptyset) \delta^P_w = ((m-2)',\{0\})$.
We have $(r,\{q\}) \delta^P_w = (r \delta'_w,\{q \delta_w\}) = ((m-1)',\{q \delta_w\})$.
Because $0$ is reachable only by the empty word in $\mathcal{D}_n$ since $L_n$ is suffix-free, we have $q \delta_w \neq 0$.
Then $((m-2)',\{0\})$ and $((m-1)',\{q \delta_w\})$ are distinguishable by our earlier considerations.

Finally, consider two states $(p_1,\emptyset)$ and $(p_2,\emptyset)$ from $R$.
There exists a word $w$ which distinguishes $p_1$ and $p_2$ in $\mathcal{D}'_m$. Let $w$ be a shortest such word.
Then, without loss of generality, $p_1 \delta'_w = (m-2)'$ and $p_2 \delta'_w \neq (m-2)'$.
Since $\mathcal{D}_m'$ is prefix-free, for every proper prefix $v$ of $w$ we have $p_1 \delta'_v \neq (m-2)'$, and since $w$ is shortest, we have $p_2 \delta'_v \neq (m-2)'$.
Then $(p_1,\emptyset) \delta^P_w = ((m-2)',\{0\})$ and $(p_2,\emptyset) \delta^P_w = (p_2 \delta'_w,\emptyset)$.
If $p_2 \delta'_w \in \{1',\ldots,(m-3)'\}$ then $((m-2)',\{0\})$ and $(p_2 \delta'_w,\emptyset)$ are distinguishable by our earlier considerations.
Otherwise $p_2 \delta'_w$ must be $(m-1)'$, and since $((m-1)',\emptyset)$ is equivalent to $((m-1)',n-1))$, it is also distinguishable from $((m-2)',\{0\})$.
\end{proof}

\subsection{Star}

The tight bound for star is $n-1$, which is met by binary bifix-free languages \cite{BrJiLiSm14QuotientComplexityOfBifixFactorSubwordFree}.
Here we provide a complete characterization for the state complexity of $L_n^*$ and show that there are exactly two possibilities for its state complexity: $n-1$ and $n-2$.
This may be compared with prefix-free languages, where there are exactly three possibilities for the state complexity $L_n^*$: $n$, $n-1$, and $n-2$ \cite{JiPaMa14}.

\begin{theorem}\label{thm:star}
Let $n \ge 3$ and let $\mathcal{D}_n=(Q,\Sigma,\delta,\{n-2\},0)$ be a minimal DFA of a bifix-free language $L_n$.
If the transformation of some $a \in \Sigma$ maps some state from $\{0,\ldots,n-3\}$ to $n-1$, then $L_n^*$ has state complexity $n-1$. Otherwise it has state complexity $n-2$.
\end{theorem}
\begin{proof}
Let $\mathcal{N}=(Q,\Sigma,\delta_\mathcal{N},\{0\},\{0\})$ be the NFA obtained from $\mathcal{D}_n$ by the standard construction for $L_n^*$:
We have $\delta^\mathcal{N}(p,a) = \{q\}$ whenever $\delta(p,a) = q$, and there is the $\varepsilon$-transition $\delta^\mathcal{N}(n-2,\varepsilon) = \{0\}$.
Let $\mathcal{D}_S=(2^Q,\Sigma,\delta^S,\{0\},\{0\} \cup 2^{Q \setminus \{0\}})$ be the DFA obtained by the powerset construction from $\mathcal{N}$.

Since in $\mathcal{D}_n$ only $n-2$ is final and the transformation of every letter maps it to empty state $n-1$, we know that only the subsets of the forms $\{q\}$, $\{q,n-1\}$, $\{0,n-2\}$, and $\{0,n-2,n-1\}$, where $q \in Q$, are reachable in $\mathcal{D}_S$.
Since $n-1$ is empty, subsets $\{q\}$ and $\{q,n-1\}$ with $q \in \{0,\ldots,n-2\}$ are equivalent, and since subsets with $0$ are final, $\{0\}$, $\{0,n-2\}$, and $\{0,n-2,n-1\}$ are equivalent.

First observe that every subset $\{q\}$ with $q \in \{0,\ldots,n-3\}$ is reachable in $\mathcal{D}_S$ by a word reaching $q$ in $\mathcal{D}_n$.
Also, if $n-1$ is reachable from some state $q \in \{0,\ldots,n-3\}$ in $\mathcal{D}_n$ by the transformation of some letter, then $\{n-1\}$ can be reached in $\mathcal{D}_S$ from $\{q\}$ by the transformation of this letter.
Otherwise, in $\mathcal{D}_S$, for all words $w$ and all subsets $S$ containing $q \in \{0,\ldots,n-3\}$ we know that $S \delta^S_w$ also contains a state from $\{0,\ldots,n-3\}$; thus subset $\{n-1\}$ cannot be reached.

Let $p,q \in \{0,\ldots,n-3,n-1\}$ be two distinct states; we will show that $\{p\}$ and $\{q\}$ are distinguishable in $\mathcal{D}_S$.
They are distinguishable in $\mathcal{D}_n$, which means that there exists a word $w \ne \varepsilon$ such that exactly one of the states $p \delta_w$ and $q \delta_w$ is final state $n-2$.
Let $w$ be a shortest such word, and without loss of generality assume that $p \delta_w = n-2$.
For any non-empty proper prefix $v$ of $w$, we have $p \delta_v \ne n-2$ because $L_n$ is prefix-free, $q \delta_v \ne 0$ because $L_n$ is suffix-free, and $q \delta_v \ne n-2$ because $w$ is shortest.
Hence, in $\mathcal{D}_S$, $\{p\} \delta^S_w = \{0,n-2\}$ and $\{q\} \delta^S_w = \{r\}$ with $r \in \{1,\ldots,n-3\}$.
Thus $w$ distinguishes both subsets.
\end{proof}

\subsection{Reversal}

For the state complexity of a reverse bifix-free language, it was shown in~\cite[Theorem 6]{BrJiLiSm14QuotientComplexityOfBifixFactorSubwordFree} that for $n \ge 3$ the tight upper bound is $2^{n-3}+2$, and that a ternary alphabet is sufficient.
We show that the alphabet size cannot be reduced, and characterize the transition semigroup of the DFAs of witness languages.

\begin{theorem}
For $n \ge 6$, to meet the bound $2^{n-3}+2$ for reversal, a witness language must have at least three letters.
Moreover, for $n \ge 5$ the transition semigroup $T(n)$ of a minimal DFA $\mathcal{D}_n(Q,\Sigma,\delta,0,\{n-2\})$ accepting a witness language must be a subsemigroup of $\Wbf(n)$.
\end{theorem}
\begin{proof}
We use the standard reversal construction: let $\mathcal{N}$ be the NFA obtained from DFA $\mathcal{D}_n$ by reversing all edges, making $n-2$ an initial state, and making $0$ a final state.
Then $\mathcal{N} = (Q,\Sigma,\delta^\mathcal{N},\{n-2\},\{0\})$, and $\delta^\mathcal{N}(q,a) = \delta^{-1}(q,a) = q \delta_a^{-1} $.
We use the powerset construction to determinize $\mathcal{N}$ to $\mathcal{D}_R$.

Remind that $Q_M = \{1,\ldots,n-3\}$.
Let $R(n)$ be the transition semigroup of $\mathcal{D}_R$.
By the reversal construction, $R(n)$ consists of all transformations $t^{-1}$ for $t \in T(n)$.
As it was shown in the proof of~\cite[Theorem~6]{BrJiLiSm14QuotientComplexityOfBifixFactorSubwordFree}, to achieve the upper bound, in particular, each subset of $Q_M$ must be reachable in $\mathcal{D}_R$.

First we show that the transition semigroup $T(n)$ of $\mathcal{D}_n$ is a subsemigroup of $\Wbf(n)$.
Every pair $\{p,q\} \subseteq Q_M$, where $p \neq q$, must be reachable in $\mathcal{D}_R$.
This means that there exists a transformation $t^{-1} \in R(n)$ such that $\{n-2\}t^{-1}=\{p,q\}$.
We have $pt=n-2$ and $qt=n-2$, so $p$ and $q$ are focused by $t \in T(n)$.
Then we know that every pair of states in $Q_M$ is focused, and since $\Wbf(n)$ is the unique maximal transition semigroup with this property (\cite{SzWi16SyntacticComplexityOfBifixFree}), we have $T(n) \subseteq \Wbf(n)$.

Now we show that a binary alphabet, say $\Sigma = \{a,b\}$, is not enough to reach the upper bound.
Let $a \in \Sigma$ be a letter such that $0 \delta_a = q \in Q_M$.
Since $T(n) \subseteq \Wbf(n)$, we have $Q_M \delta_a \subseteq \{n-2,n-1\}$.
Moreover, for all $p \in Q_M$ we have $\{p\}\delta_a^{-1} \subseteq \{0\}$.
Also, the set $\{n-2\}\delta^{-1}_b$ is empty, because if it contains some state $p \in Q_M$, then at most two of the $2^{n-4}$ subsets of $Q_M$ containing $p$ might be reachable, since $p \delta_b = \{n-2\}$ and $p \delta_a \in \{n-2,n-1\}$; then not all subsets of $Q_M$ are reachable since $n \ge 6$.
Because $\{n-2,n-1\}\delta_a = \{n-2,n-1\}\delta_b = \{n-1\}$, any subset containing $n-2$ is reachable in $\mathcal{D}_R$ only by the empty word, and any subset containing $n-1$ is unreachable.
Hence, a non-empty subset $S$ of $Q_M$ must be reachable in $\mathcal{D}_R$ by a word of the form $a b^i$, that is, $S = \{n-2\}\delta_{b^i a}^{-1}$.
Let $C$ be the states from $Q_M$ that are fixed points or belong to a cycle in $\delta_b$.
Observe that $\delta_b^{-1}$ does not change the cardinality in $C$, that is, for any subset $T \subseteq Q$ we have $|T \cap C| = |T \delta_b^{-1} \cap C|$.
Thus, if $C$ is non-empty, only subsets with the same cardinality in $C$ are reachable.
If $C$ is empty, then $S \delta_{b^{n-3}}^{-1} \cap Q_M = \emptyset$, thus at most $n-2$ subsets of $Q_M$ are reachable; since $n \ge 6$, not all subsets of $Q_M$ are reachable.
\end{proof}

It is known that in the case of the class of all regular languages the resulted language of the reversal operation can have any state complexity in range of integers $[\log_2 n,2^n]$ \cite{Jir2008,Seb2013}, thus there are no gaps (\emph{magic numbers}) in the interval of possible state complexities.
The next theorem states that the situation is similar for the case of bifix-free languages.

\begin{theorem}\label{thm:reversal_magic}
If $L_n$ is a bifix-free language with state complexity $n \ge 3$, then the state complexity of $L_n^R$ is in $[3+\log_2(n-2),2+2^{n-3}]$.
Moreover, all values in this range are attainable by $L_n^R$ for some bifix-free language $L_n$, whose a minimal DFA has transition semigroup that is a subsemigroup of $\Wbf(n)$.
\end{theorem}
\begin{proof}
Note that if $L_n$ is a bifix-free language, then so $L_n^R$ is.
Also we know that $(L_n^R)^R = L_n$.

Suppose for a contradiction that there is some $L_n$ whose $L_n^R$ has state complexity $\alpha < 3+\log_2(n-2)$.
We have $L_n^R$ with state complexity $\alpha$ whose reverse $(L_n^R)^R = L_n$ has state complexity $n$.
However, since $n > 2^{\alpha-3}+2$, this means that $(L_n^R)^R$ exceeds the upper bound for reversal.

Now it is enough to show that every value from $[n,2+2^{n-3}]$ is attainable, because to reach $\alpha \in [3+\log_2(n-2),n-1]$ for some $L_n^R$, we can use $L_\alpha = L_n^R$ whose reverse $(L_n^R)^R = L_n$ has state complexity $n \in [\alpha+1,2+2^{\alpha-3}]$.

Let $\alpha \in [n,2+2^{n-3}]$.
We construct the DFA $\mathcal{D}_n(Q,\Sigma,\delta,0,\{n-2\})$ as follows:
For each $q \in \{1,\ldots,n-3\}$, we put the letter $a_q\colon (0 \to q)(\{1,\ldots,n-2 \to n-1)$.
We choose $\alpha-2$ subsets of $\{1,\ldots,n-3\}$ in such a way that we always have $\emptyset$ and all singletons $\{q\}$, while the other $\alpha-3-(n-3)$ subsets are chosen arbitrarily.
Since $\alpha \le 2+2^{n-3}$, $n \ge 3$, and the number of the subsets is $2^{n-3}$, we can always make that choice.
For each of the chosen subset, we put the letter $b_S\colon (S \to n-2)(Q \setminus S \to n-1)$.

Observe that $\mathcal{D}_n$ is minimal: every state $q \in \{1,\ldots,n-3\}$ is reachable by $a_q$, and $n-2$ and $n-1$ are reachable using some $b_S$.
State $0$ is distinguished since $a_q$ maps all the other states to empty state $n-1$. Two distinct states $p,q \in \{1,\ldots,n-3\}$ are distinguished by $b_{\{p\}}$, since its action maps $p$ to $n-2$ and $q$ to $n-1$.
Also, the transformations of $a_q$ and $b_S$ are of type~3 and type~1 from Definition~\ref{def:Wbf}, and so the transition semigroup of $\mathcal{D}_n$ is a subsemigroup of $\Wbf(n)$.

Let $\mathcal{D}_R(2^Q,\Sigma,\delta^{-1},\{n-2\},\{0\} \cup 2^{Q \setminus \{0\}})$ be the automaton recognizing $L_n^R$ obtained by reversing the edges of $\mathcal{D}_n$ and determinization.
We show that there are exactly $\alpha$ reachable subsets in $\mathcal{D}_R$.
The transformation $\delta^{-1}_{b_S}$ allows to reach the subset $S$ from $\{n-2\}$.
Then for any non-empty $S$, $\delta^{-1}_{a_q}$ for some $q \in S$ allows to reach $\{0\}$ from $S$.
Thus we have reachable $\{0\}$, $\{n-2\}$, and $\alpha-2$ chosen subsets $S$.
No subset containing $n-1$ can be reached, and $\{n-2\}$ is the only reachable subset containing $n-2$.
Since $\{0\}$ is reachable only by the transformations $\delta^{-1}_{a_q}$, and these transformations map $Q \setminus \{q\}$ to $\emptyset$, no other subset containing $0$ can be reached.
Since the transformations $\delta^{-1}_{a_q}$ of the letters $a_q$ do not map any state to a state in $\{1,\ldots,n-3\}$, to reach $S \subseteq \{1,\ldots,n-3\}$ we must use the transformations of the letters of type $b_S$. But this is possible only from $\{n-2\}$ for the chosen $\alpha-2$ subsets.

Finally, observe that every two distinct subsets $S,S' \subseteq \{1,\ldots,n-3\}$ are distinguished by $a_q^{-1}$ such that $q \in S \oplus S'$.
Subset $\{0\}$ is the final subset, and $\{n-2\}$ is distinguished as it is the only subset for which $b_{\{q\}}$ does not result in $\emptyset$.
\end{proof}

\subsection{Atom complexities}\label{subsec:atoms}

Here we prove a tight upper bound on the number and the state complexities of atoms of a bifix-free language.

Remind that for $S \subseteq Q$ an atom $A_S = \bigcap_{q \in S} K_q \cap \bigcap_{q \in Q \setminus S} \overline{K_q}$ is a non-empty set.
Then for any $w\in\Sigma^*$ we have
$$A_S.w =  \{u \mid wu \in A_S\} = \bigcap_{q \in S}K_q.w \cap \bigcap_{q \in \overline{S}} \overline{K_q.w}.$$
A quotient of a quotient of $L$ is also a quotient of $L$, and therefore $A_S.w$ has the following form:
$$A_S.w = \bigcap_{q \in X} K_q \cap \bigcap_{q \in Y} \overline{K_q},$$
where $|X| \le |S|$, $|Y| \le n-|S|$, and $X,Y \subseteq Q$ are disjoint.

Using the approach from Iv\'an~\cite{Iva16ComplexityOfAtoms} we define the DFA $\mathcal{D}_S=(Q_S,\Sigma,\delta,(S,Q \setminus S),F_S)$ such that:
\begin{itemize}
\item $Q_S = \{(X,Y) \mid X,Y \subseteq Q, X \cap Y = \emptyset\} \cup \{\bot\}$.
\item For all $a \in \Sigma$, $(X,Y)a = (Xa,Ya)$ if $X t_a \cap Y t_a = \emptyset$, and $(X,Y)a = \bot$ otherwise; also $\bot t_a = \bot$.
\item $F_S = \{(X,Y) \mid X\subseteq \{n-2\}, Y \subseteq Q \setminus \{n-2\}\}$. 
\end{itemize}

Then DFA $\mathcal{D}_S$ recognizes $A_S$, and so if $\mathcal{D}_S$ recognizes a non-empty language, then $A_S$ is an atom. 
Every quotient of an atom is represented by a pair $(X,Y)$.

\begin{theorem}\label{thm:atoms_upper_bound}
Suppose that $L_n$ is a bifix-free language recognized by a minimal DFA $\mathcal{D}_n(Q,\Sigma,\delta,0,\{n-2\})$.
Then there are at most $2^{n-3}+2$ atoms of $L_n$ and the quotient complexity of $\kappa(A_S)$ of atom $A_S$ satisfies:
\[
\kappa(A_S) \begin{cases}
\le 2^{n-2}+1 & \text{if $S = \emptyset$;}\\
= n & \text{if $S = \{0\}$;}\\
= 2 & \text{if $S = \{n-2\}$;} \\
\le 3 + \sum_{x=1}^{|S|}\sum_{y=0}^{n-3-|S|}\binom{n-3}{x}\binom{n-3-x}{y} &\text{if $\emptyset \neq S\subseteq \{1,\ldots,n-3\}$.}
\end{cases}
\]
\end{theorem}
\begin{proof}
We follow similarly to the proof from~\cite{BrSz15ComplexityOfSuffixFree} for the class of suffix-free languages.

If $n-1 \in S$ then $A_S$ would be empty, because the quotient of $n-1$ is the empty language, and so will not form an atom.

Suppose that $0 \in S$. Since $L_n$ is suffix-free, we have $K_0 \cap K_q = \emptyset$ for $q \neq 0$, so if $\{0,q\} \subseteq S$ then $A_S$ would be empty. Thus $S = \{0\}$, $A_0 = K_0$, and the quotient of $0$ has complexity $n$, since $\mathcal{D}_n$ is minimal.

Suppose that $n-2 \in S$. Since $K_{n-2} = \{\varepsilon\}$ and this is the quotient containing the empty word, we have $K_{n-2} \cap K_q = \emptyset$ for $q \neq n-2$, so if $\{n-2,q\} \subseteq S$ then $A_S$ would be empty.
Thus $S = \{n-2\}$, and the quotient of $n-2$ has complexity $2$.

It follows that there are at most $2^{n-3}+2$ atoms: $A_{\{0\}}$, $A_{\{n-2\}}$, and $A_S$ for any $S \subseteq \{1,\ldots,n-3\}$.

Suppose that $S = \emptyset$.
Then $A_\emptyset = \bigcap_{q \in Q} \overline{K_q}$.
For any $w \in \Sigma^+$ we have $A_\emptyset.w = \bigcap_{q \in Y} \overline{K_q}$, for some $Y \subseteq Q \setminus \{0\}$.
We can assume that $n-1 \in Y$ since $\overline{K_{n-1}} = \Sigma^*$.
Thus there are at most $2^{n-2}$ choices of $Y$, which together with the initial quotient $A_\emptyset$ yields the quotient complexity $2^{n-2}+1$.

Suppose that $\emptyset \neq S \subseteq \{1,\ldots,n-3\}$.
Consider the non-empty quotient $A_S.w$ for some $w \in \Sigma^+$ be represented as $\bigcap_{q \in X} K_q \cap \bigcap_{q \in Y} \overline{K_q}$.
So $X$ has at least one and at most $|S|$ states from $Q \setminus \{0\}$,
$Y \subseteq Q \setminus \{0\}$ and always contains $n-1$.
If $n-2 \in X$ then if this quotient is non-empty then it must be $\{\varepsilon\}$, and if $n-2 \in Y$ then the quotient may be represented by $(X,Y \setminus \{n-2\})$.

If $0 \delta_w \in \{n-2,n-1\}$ then $Y \setminus \{n-2,n-1\}$ contains from $0$ to at most $n-3-|S|$ states from $Q \setminus (\{0,n-2,n-1\} \cup X)$.
Suppose that $0 \delta_w = q \in Q_M$. Since the language is suffix-free, the path in $\delta_w$ starting at $0$ must end in $n-1$, as otherwise $0 \delta^i_w = q \delta^i_w \neq n-1$ for some $i$, which contradicts Lemma~\ref{lem:bifix-free_properties}(Condition~3).
But then there exists a state $p \in Q_M$ such that $p \delta_w \in \{n-2,n-1\}$.
If $p \in S$, then $n-1 \in X$, and so $(X,Y)$ represents the empty quotient.
If $p \in Q \setminus S$, then again $Y \setminus \{n-2,n-1\}$ contains at most $n-3-|S|$ states.

So for every choice of $X$ we have $0 \le |Y \setminus \{n-2,n-1\}| \le n-3-|S|$ from $Q \setminus (\{0,n-2,n-1\} \cup X)$, which together with the initial quotient, $\{\varepsilon\}$ quotient, and the empty quotient yields the formula in the theorem.
\end{proof}

\begin{theorem}\label{thm:atoms_lower_bound}
For $n \ge 6$, let $L_n$ be the language recognized by the DFA $\mathcal{D}(Q,\Sigma,\delta,0,\{n-2\})$, where $\Sigma = \{a,b,c,d,e_1\ldots,e_{n-3}\}$, and $\delta$ is defined as follows:\\
$\delta_a\colon (0 \to 1)((Q \setminus \{0\}) \to n-1)$,\\
$\delta_b\colon (\{0,n-2\} \to n-1)(1,2)$,\\
$\delta_c\colon (\{0,n-2\} \to n-1)(1,\ldots,n-3)$,\\
$\delta_d\colon (\{0,n-2\} \to n-1)(2 \to 1)$,\\
$\delta_{e_q}\colon (\{0,n-2\} \to n-1)(q \to n-2)$ for $q \in Q_M$.\\
Then $\mathcal{D}$ is minimal, $L_n$ is bifix-free and it meets the upper bounds for the number and complexities of atoms from Theorem~\ref{thm:atoms_upper_bound}.
\end{theorem}
\begin{proof}
It is easy to observe that $\mathcal{D}$ is minimal, it recognizes a bifix-free language $L_n$, and its transition semigroup is a subsemigroup of $\Wbf(n)$.
We show that the atom complexities $\kappa(A_S)$ meet the bounds, which also implies that there are $2^{n-3}+2$ atoms.

Observe that in the transition semigroup of $\mathcal{D}$ we have all transformations of type~1 from Definition~\ref{def:Wbf}: the transformations of $b$, $c$, and $d$ generate all transformations on $Q_M$ (\cite{Pic38}) with $(\{0,n-2\} \to n-1)$, using $e^2_q$ we can map any state $q$ to $n-1$, and finally by using $e_q$ we can map any state $q$ to $n-2$.

First, by the properties of prefix-free and suffix-free languages (see Lemma~\ref{lem:bifix-free_properties}), we have the two special atoms $A_{\{n-2\}} = \{\varepsilon\}$ and $A_{\{0\}} = L_n$, which obviously meet the bounds $2$ and $n$.
Now consider $S \subseteq \{1,\ldots,n-3\}$ and the automaton $\mathcal{D}_S$ recognizing $A_S$.

For $S = \emptyset$ we know that the initial state of $\mathcal{D}_S$ is $(\emptyset,Q)$.
Since we have all transformations of type~1, we can reach $(\emptyset,Y)$ for all $Y \subseteq Q \setminus \{0\}$ which contains $n-1$.
Therefore, $2^{n-2}$ states and the initial state are reachable.
All states $(\emptyset,Y_1)$ and $(\emptyset,Y_2)$ for $Y_1,Y_2 \subseteq Q \setminus \{0\}$ which contain $n-1$ are distinguishable as follows.
If $q \in Y_1$ and $q \notin Y_2$, then we can use $(q \to n-2)((Q \setminus \{q\}) \to n-1)$ of type~1, which accepts from $(\emptyset,Y_2)$ but does not accept from $(\emptyset,Y_1)$.
Also $(\emptyset,Q)$ is distinguished from the other states by $\delta_a\delta_{e_1}$.

Now consider $\emptyset \neq S \subseteq \{1,\ldots,n-3\}$.
We show that all states $(X,Y)$ such that $X \subseteq Q_M$, $n-1 \in Y \subseteq (Q_M \setminus X) \cup \{n-1\}$, $|X| \le |S|$, $|Y| \le n-2-|S|$, together with the initial state $(S,Q \setminus S)$, the empty state $\bot$, and a final state $(\{n-2\},\{n-1\})$, are reachable and pairwise distinguishable.
This yields the number of states in the formula.

We can map $S$ to $X$ and $Q \setminus S$ to $Y$ by a transformation of type~1 from $\Wbf(n)$. Therefore, all states $(X,Y)$ are reachable.
Also by a transformation of type~1 we can map $S$ to $n-2$, and $Q \setminus S$ to $n-1$, so $(\{n-2\},\{n-1\})$ is reachable.
Then the empty state is reachable from $(\{n-2\},\{n-1\})$ by any transformation.

State $(\{n-2\},\{n-1\})$ is the only final state which we consider.
A state $(X,Y)$ is non-empty since we can map $X$ to $n-2$ and $Y$ to $n-1$ by a transformation of type~1.
Consider two states $(X_1,Y_1) \neq (X_2,Y_2)$.
Without loss of generality, if $q \in X_1$ and $q \notin X_2$, then the transformation $(X_1 \to n-2)((Q \setminus X_1) \to n-1)$ of type~1 distinguishes both pairs.
Otherwise, if $q \in Y_1$ and $q \notin Y_2$, then the transformation $(0 \to n-1)(q \to x)(n-2 \to n-1)$ of type~1, where $x$ is some state from $X_1$, yields the empty state from $(X_1,Y_1)$, but the non-empty state $(X_2,Y_2 \setminus \{0,n-2\})$ from $(X_2,Y_2)$.
\end{proof}

\begin{theorem}
For $n \ge 7$, to meet the upper bounds for the atom complexities from Theorem~\ref{thm:atoms_upper_bound} by the language of a minimal DFA $\mathcal{D}_n(Q,\Sigma,\delta,0,\{n-2\})$, the size of the alphabet $\Sigma$ must be at least $n+1$.
Moreover, the transition semigroup of $\mathcal{D}_n$ must be a subsemigroup of $\Wbf(n)$.
\end{theorem}
\begin{proof}
First we show that the transition semigroup of $\mathcal{D}_n$ is a subsemigroup of $\Wbf(n)$.
Consider the atom $A_S$ for $S = Q_M$, and the DFA $\mathcal{D}_S$.
In particular the state $(\{1\},\{n-1\})$ must be reachable.
So there is a transformation mapping $Q_M$ to $1$, which means that all colliding pairs are focused.
Since $\Wbf(n)$ is the unique maximal transition semigroup with this property (\cite{SzWi16SyntacticComplexityOfBifixFree}), the transition semigroup of $\mathcal{D}_n$ must be a subsemigroup of it.

Now we show that to meet the bounds we require at least $n$ letters.
Since $\mathcal{D}$ is minimal, there must be a letter of type~3, which maps $0$ to a state from $Q_M$; let it be named $a$.

Consider the atom $A_\emptyset$.
Then all states $(\emptyset,Y)$ for $n-1 \in Y \subseteq Q \setminus \{0\}$ must be reachable.
In particular, for every $q$ the state $(\emptyset,(Q_M \setminus \{q\}) \cup \{n-2,n-1\})$ is reachable.
Let it be reachable by some word $w$, and consider the transformation $t$ of the last letter in this word.
Since the transition semigroup is a subsemigroup of $\Wbf(n)$, $t$ must be of type~1, so it maps $0$ to $n-1$, $q$ to $n-2$, and permutes the other states from $Q_M$.
Therefore, we must have $n-3$ different such transformations for every $q \in Q_M$; let their letters be named $e_q$.

Consider the atom $A_S$ for $S = Q_M$.
Since a state $(S \setminus \{q\},\{n-1\})$ for any $q \in Q_M$ must be reachable, in the transition semigroup there must be a transformation that focuses a pair of states from $Q_M$ and does not map any state from $Q_M$ to $n-2$ nor to $n-1$.
If this transformation is induced by some word, then there must be a letter in this word with the same property. So it must be different from $a$ and from $\delta_{e_q}$; let it be named $d$.

Finally, we show that there must be at least two other letters inducing transformations of type~1 acts on $Q_M$ as permutations.
Note that these transformations cannot be generated from $\delta_a$, $\delta_d$, and the transformations $\delta_{e_q}$.

Consider the atom $A_S$ for $S = \{1,2\}$.
Then, in particular, all $\binom{n-3}{2}$ states $(X,Y)$ with $|X| = 2$ and $Y =  (Q_M \setminus X) \cup \{n-1\}$ must be reachable in $\mathcal{D}_S$, which means that there is a letter inducing permutation on $Q_M$.
Let it be named $c$.
If there are no more letters than $a$, $c$, $d$, and $e_{q}$, then all states $(X,Y)$ with $|X| = 2$ and $Y =  (Q_M \setminus X) \cup \{n-1\}$ must be reachable from the initial $(\{1,2\},Q \setminus \{1,2\})$ just by $c^i$, since the transformations of the other letters are not permutations on $Q_M$.

We can decompose $\delta_c$ restricted to $Q_M$ into disjoint cycles.
Suppose that there are two or more cycles, on $C_1,C_2 \subseteq Q_M$, respectively.
Then, however, only states $(X,Y)$ with the same number $|X \cap C_1|=|\{1,2\} \cap C_1|$ can be reached by $c^i$, as $\delta_c$ preserves this number.
Therefore, $\delta_c$ must be a cycle on $Q_M$.
Then, however, at most $n-3$ states $(X,Y)$ can be reached, which is less than $\binom{n-3}{2}$ for $n \ge 7$.
Thus, there must be another letter (inducing a permutation on $Q_M$), say $b$.

Summarizing, we have the letters $a$, $b$, $c$, $d$, and $n-3$ letters $e_q$ inducing different transformations.
\end{proof}


\section{Most complex bifix-free languages}

\subsection{A ternary most complex stream for basic operations}

First we show a most complex stream of bifix-free languages for basic operations which uses only a ternary alphabet.
This alphabet is a smallest possible, because for union, symmetric difference, and reversal we require at least three letters to meet the bounds.

\begin{definition}[Most complex stream for operations]\label{def:most_complex_operations}
For $n \geq 7$, we define the DFA $\mathcal{D}_n = ( Q,\Sigma,\delta,0,\{n-2\})$, where $ Q = \{0,\dots,n-1\}$, $\Sigma = \{a,b,c\}$, $h=\lfloor (n-1)/2 \rfloor$ and $\delta$ is defined as follows:
\begin{itemize}
\item $\delta_a\colon (0 \rightarrow 1)(\{1,\dots,n-3\} \rightarrow n-2)(\{n-2,n-1\} \rightarrow n-1),$
\item $\delta_b\colon (\{0,n-2,n-1\} \rightarrow n-1)(1,\dots,n-3),$
\item $\delta_c\colon (\{0,n-2,n-1\} \rightarrow n-1)(1 \rightarrow h)(h \rightarrow n-2)(n-3,\dots,h+1,h-1,\dots,2).$
\end{itemize}
The DFA $\mathcal{D}_n$ is illustrated in Fig.~\ref{fig:most_complex_operations}.
\end{definition}

\begin{figure}[ht]
\unitlength 14pt\small
\gasset{Nh=2,Nw=2,Nmr=1,ELdist=0.3,loopdiam=1.5}
\begin{center}\begin{picture}(27,22)(-1,-1)
\node(1)(6,10){$1$}
\node(2)(8.5,17.5){$2$}
\node[Nframe=n](dots1)(16,20){$\dots$}
\node(h-1)(23.5,17.5){$h-1$}
\node(h)(26,10){$h$}
\node(h+1)(23.5,2.5){$h+1$}
\node[Nframe=n](dots2)(16,0){$\dots$}
\node(n-3)(8.5,2.5){$n-3$}
\node(0)(0,10){$0$}\imark(0)
\node(n-2)(16,10){$n-2$}\rmark(n-2)
\drawedge(0,1){$a$}
\drawedge[curvedepth=0,ELpos=40](1,n-2){$a$}
\drawedge[curvedepth=0,ELpos=40](2,n-2){$a$}
\drawedge[curvedepth=0,ELpos=40](dots1,n-2){$a$}
\drawedge[curvedepth=0,ELpos=40](h-1,n-2){$a$}
\drawedge[curvedepth=0,ELpos=40](h,n-2){$a,c$}
\drawedge[curvedepth=0,ELpos=40](h+1,n-2){$a$}
\drawedge[curvedepth=0,ELpos=40](dots2,n-2){$a$}
\drawedge[curvedepth=0,ELpos=40](n-3,n-2){$a$}
\drawedge[curvedepth=-.7,ELside=r](1,2){$b$}
\drawedge[curvedepth=-.7,ELside=r](2,dots1){$b$}
\drawedge[curvedepth=-.7,ELside=r](dots1,2){$c$}
\drawedge[curvedepth=-.7,ELside=r](dots1,h-1){$b$}
\drawedge[curvedepth=-.7,ELside=r](h-1,dots1){$c$}
\drawedge[curvedepth=-.7,ELside=r](h-1,h){$b$}
\drawedge[curvedepth=-.7,ELside=r](h,h+1){$b$}
\drawedge[curvedepth=-.7,ELside=r](h+1,dots2){$b$}
\drawedge[curvedepth=-.7,ELside=r](dots2,h+1){$c$}
\drawedge[curvedepth=-.7,ELside=r](dots2,n-3){$b$}
\drawedge[curvedepth=-.7,ELside=r](n-3,dots2){$c$}
\drawedge[curvedepth=-.7,ELside=r](n-3,1){$b$}
\drawedge[curvedepth=-4.2,ELside=r,ELpos=30](h+1,h-1){$c$}
\drawedge[curvedepth=-4.2,ELside=r,ELpos=30](2,n-3){$c$}
\drawedge[curvedepth=-3,ELpos=25](1,h){$c$}
\end{picture}\end{center}
\caption{Automaton $\mathcal{D}_n$ from Definition~\ref{def:most_complex_operations}. Empty state $n-1$ and the transitions going $n-1$ are omitted.}\label{fig:most_complex_operations}
\end{figure}
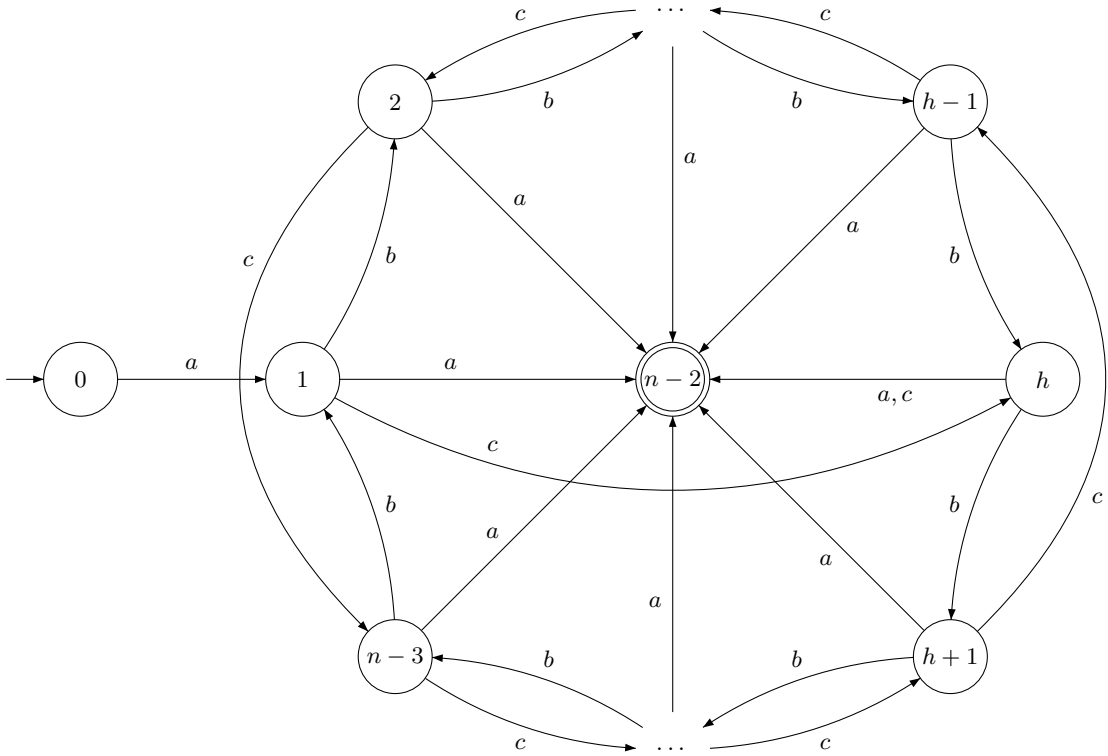

\begin{theorem}\label{thm:most_complex_operations}
The DFA $\mathcal{D}_n$ from Definition~\ref{def:most_complex_operations} is minimal, recognizes a bifix-free language $L_n(a,b,c)$, has most complex quotients, and its transition semigroup is a subsemigroup of $\Wbf(n)$.
The stream $(L_n(a,b,c) \mid n \ge 9)$ with some permutationally equivalent dialects meets all the bounds for basic operations as follows:
\begin{itemize}
\item $L_m(a,b,c)$ and $L_n(a,c,b)$ meets the bound $mn-(m+n)$ for union and symmetric difference, the bound $mn-3(m+n-4)$ for intersection and the bound $mn-(2m+3n-9)$ for difference.
\item $L_m(a,b,c)$ and $L_m(a,b,c)$ meets the bound $m+n-2$ for product.
\item $L_m(a,b,c)$ meets the bound $n-1$ for star.
\item $L_m(a,b,c)$ meets the bound $2^{n-3}+2$ for reversal.
\end{itemize}
\end{theorem}

\begin{proof}

\noindent\textbf{Minimality}:

First we need to show that $\mathcal{D}_n$ is minimal.
For every state Table~\ref{tab:distinguishing_words} depicts a word by which that state is reached, and a word that is accepted from this state and is not accepted from any other state, thus shows that all states are reachable and distinguishable.

\renewcommand{\arraystretch}{1.2}
{\small\begin{table}[ht]\centering
\caption{Reaching and distinguishing words for stream from Definition ~\ref{def:most_complex_operations} }\label{tab:distinguishing_words}
\begin{tabular}{|l|c|c|}\hline
State        & Reachable by  & Distinguishing accepted word\\ \hline\hline
$0$          & $\varepsilon$ & $ac^2$   \\ \hline
$ q \in Q_M$ & $ab^{q-1}$    & $b^{n-2-q}c^2$   \\ \hline
$n-2$        & $a^2$         & $\varepsilon$   \\ \hline
$n-1$        & $a^3$         & no word exists   \\ \hline
\end{tabular}
\end{table}}

\noindent\textbf{Most complex quotients}:

Since from every state $q \in \{1,\ldots,n-3\}$ we can reach any ostate from $Q \setminus \{0\}$, their quotients have state complexity $n-1$, and so the language of $\mathcal{D}_n$ meets the bound from Proposition~\ref{pro:quotients_upper_bound}.

\noindent\textbf{Subsemigroup of $\Wbf(n)$}:

Consider a non-empty word $w$ and its induced transformation $\delta_w$ from the transition semigroup $T(n)$ of $\mathcal{D}_n$.
We will show that $\delta_w \in \Wbf(n)$.
If $0\delta_w = n-1$ then $\delta_w$ is of type~1, since there are no transitions going to state $0$.
If $0\delta_w = n-2$ then $w$ starts with $a$ and has length at least $2$, so $q\delta_w = n-1$ for all $q \in Q_M$, and so $\delta_w$ is of type~2.
If $0\delta_w \in Q_M$ then also $w$ starts with $a$, so $q\delta_w \in \{n-2,n-1\}$ for all $q \in Q_M$, and so $\delta_w$ is of type~3.

\noindent\textbf{Recognizing a bifix-free language}:

It is enough to observe that $\mathcal{D}_n$ satisfies the three sufficient and necessary conditions from Lemma~\ref{lem:bifix-free_properties}.
Conditions~1 and~2 are trivially satisfied by states $n-1$ and $n-2$.
For condition~3 note that all words inducing a transformation in $\Wbf(n)$ and so in $T(n)$ satisfy it.

\noindent\textbf{Product and star:}

From Theorem~\ref{thm:product} we know that $L_m(a,b,c)L_n(a,b,c)$ meets the bound for product.
Since transformation $\delta_b$ maps state $0$ to state $n-1$, by Theorem~\ref{thm:star} we know that $L^*_n(a,b,c)$ meets the bound $n-1$ for star.

\noindent\textbf{Reversal:}

Consider the standard construction of the NFA $\mathcal{N}$ obtained by reversing the edges in $\mathcal{D}_n$ and its determinization to $\mathcal{D}_R$, which recognizes $L^R_n(a,b,c)$. Hence, every word $w$ in $\mathcal{D}_R$ induces a transformation acting on $2^Q$ as $\delta_w^{-1}$.

First we show that subsets (states of $\mathcal{D}_R$) $\{n-2\}$, $\{0\}$, and each of the subsets of $Q_M$ is reachable in $\mathcal{D}_R$.
Subset $\{n-2\}$ is reachable as it is the initial state of $\mathcal{D}_R$, and we have $\{n-2\} \delta_a^{-1} = Q_M$ and $Q_M \delta_a^{-1} = \{0\}$.
We show by reverse induction on cardinality that every subset of $Q_M$ is reachable.
Suppose that any subset of $Q_M$ of size $k$ is reachable.

Consider a subset $S \subset Q_M$ of size $k-1 < n-3$.
Since $S$ is a proper subset of $Q_M$, it is reachable from a subset $S'$ of size $k$ such that $h \notin S'$ by the rotation $(\delta_b^{-1})^i$ for some $i \in \{0,\ldots,n-4\}$.
Let $S'' = \{1\} \cup S' \delta_c$.
We have $\{1\} \delta_c^{-1} = \emptyset$.
Also, since $\delta_c^{-1}$ acts like like a cycle on $Q_M \setminus \{1,h\}$, $h \notin S'$, and $1 \delta_c = h$, we know that $S' \delta_c$ has size $|S'|$, does not contain $1$, and contains $h$ if and only if $S'$ contains $1$.
Then also $(S' \delta_c)\delta_c^{-1} = S'$.
Hence, $S''$ has size $|S|+1 = k$, and the induction step holds.

It remains to show that all the listed subsets are pairwise distinguishable.
Only $\{n-2\}$ accepts $a^2$, and only $\{0\}$ is final.
Consider any two different subsets $S,T \subseteq Q_M$.
Without loss of generality let $q \in S \setminus T$.
Then $\mathcal{D}_R$ from $S$ accepts the word $b^{q-1}a$ whereas from $T$ it does not.

\noindent\textbf{Reachability for Boolean operations:}

Let $\mathcal{D}'_m(Q',\Sigma',\delta',0',\{(m-2)'\})$, where $Q'=\{0',\ldots,(m-1)'\}$, be $\mathcal{D}_m$ of $L_m(a,b,c)$ with the states renamed, and let $\mathcal{E}_n(Q,\Sigma,\delta'',0,\{n-2\})$ be the DFA of the dialect $L_n(a,b,c)$ obtained from $\mathcal{D}_n$ by swapping the transitions of letters $b$ and $c$.
Let $\mathcal{D}_P(Q' \times Q,\Sigma,\delta^P,(0',0),F^P)$ be the direct product automaton of $\mathcal{D}'_m$ and $\mathcal{E}_n$.

The set of final states $F^P$ depends on the particular operation and does not matter for reachability.
Therefore, we have $(p',q) \delta^P_a = (p' \delta'_a, q \delta''_a) = (p' \delta'_a, q \delta_a)$, $(p',q) \delta^P_b = (p' \delta'_b, q \delta''_c) = (p' \delta'_b, q \delta_b)$, and $(p',q) \delta^P_c = (p' \delta'_c, q \delta''_b) = (p' \delta'_c, q \delta_c)$.
We will show that all $mn-(m+n)+2$ states in $\mathcal{D}_P$ are reachable, namely the initial state $(0',0)$ and all $(p',q)$ for $p' \in Q' \setminus \{0'\}$, $q \in Q \setminus \{0\}$.

Consider a pair $(p',q)$ with $p' \in Q_M'$ and $q \in Q_M$.
Let $Q_C'=Q_M' \setminus \{1',h'\}$ and $Q_C=Q_M \setminus \{1,h\}$.
First we show that any state of the form $(2',q)$ with $q \in Q_M$ is reachable.
Then we show as a consequence of the previous fact that any state from the set $Q_C' \times Q_M \cup Q_M' \times Q_C$ is reachable.
In the third step we show that the states in $\{1',h'\}\times \{1,h\}$ are also reachable.

State $(2',h)$ can be reached from $(0',0)$ by the word $ab$.
Consider the following two transformations that fix state $2'$ in $\mathcal{D}'_m$: $cb^2$ and $bc$.
Note that states $(2',h-2),(2',h-3),\ldots,(2',3),(2',2),(2',n-3)$ are reachable from $(2',h)$ by the words $(cb^2)^i$ for $i=1,\ldots,h-2$.
Moreover, for $n \ge 9$, since $h \ge 4$, we have $(2',2)\delta^P_{bc}=(2',1)$.
Next, from $(2',1)$ states $(2',n-4),(2',n-5),\ldots,(2',h+2),(2',h+1),(2',h-1)$ are reachable by the words $(cb^2)^i$ for $i=1,\ldots,n-3-h$.
So we have shown that all states from $\{2'\} \times Q_M$ are reachable.

In $\mathcal{D}_n$, transformation $\delta_b$ restricted to $Q_M$ and transformation $\delta_c$ restricted to $Q_C$ are cycles, and so we can map any state from $Q_M$ to any other state from $Q_M$ by $\delta_b^i$ for some $i$, and similarly for $Q_C$ with $\delta_c^i$.
Now we show that states $(p',q) \in Q_M' \times Q_C$ are reachable in $\mathcal{D}^P$.
Let $i$ be such that $2' \delta'_{b^i} = p'$ and let $r \in Q_M$ be such that $r \delta_{c^i} = q$.
Then we have $(2',r) \delta^P_{b^i} = (2' \delta'_{b^i}, r \delta_{c^i}) = (p',q)$, thus $(p',q)$ is reachable.
By the symmetric argument, states $(p',q) \in Q_C' \times Q_M$ are also reachable.

Table~\ref{tab:four_states_reachablility} shows how to reach the four special states in $\{1',h'\} \times \{1,h\}$.

\renewcommand{\arraystretch}{1.2}
{\small\begin{table}[ht]\centering
\caption{Reachability of states in $\{1',h'\} \times \{1,h\}$}.\label{tab:four_states_reachablility}
\begin{tabular}{|l|c|c|}\hline
State     & Reachable from & By word \\ \hline\hline
$(1',1)$  & $(0',0)$       & $a$     \\ \hline
$(1',h)$  & $((m-3)',1)$   & $b$     \\ \hline
$(h',1)$  & $(1',n-3)$     & $c$     \\ \hline
$(h',h')$ & $(1',1)$       & $(bc)^2$ \\ \hline
\end{tabular}
\end{table}}

Finally we show that the states $(p',q)$ where $p' \in \{(m-2)',(m-1)'\}$ or $q \in \{n-2,n-1\}$ are reachable.
State $((m-2)',n-2)$ is reachable from $(0',0)$ by $a^2$, state $((m-2)',n-1)$ by $a b^2 a$, state $((m-1)',n-2)$ by $a c^2 a$, and state $((m-1)',n-1)$ by $a^3$.
Now, using symmetry, without loss of generality, we can assume that $p' \in \{(m-2)',(m-1)'\}$ and $q \in Q_M$.
If $p' = (m-2)'$, then $((m-2)',q)$ is reachable by $c$ from $(h',q-1)$ or from $(h',n-3)$ if $q=1$.
If $p' = (m-1)'$, then similarly $((m-1)',q)$ is reachable by $c$ from $((m-2)',q-1)$, or from $((m-2)',n-3)$ if $q=1$.

\noindent\textbf{Union}:

The union is recognized by the product automaton $\mathcal{D}^P$ with $F^P = (\{(m-2)'\} \times Q) \cup (Q'_m \times \{n-2\})$.

States $((m-2)',n-2)$, $((m-2)',n-1)$, and $((m-1)',n-2)$ have the same quotients, and therefore are not distinguishable.
We will show that all the remaining reachable states together with one of the above are pairwise distinguishable.
These will be $mn-(m+n)$ states.
Consider two distinct states $(p_1',q_1)$ and $(p_2',q_2)$.
Since union is symmetric, without loss of generality we can assume that $p'_1 \neq p'_2$.

If $p_1' \in Q'_M$, then we can use the distinguishing word $b^{m-2-p_1}c^2$ from Table~\ref{tab:distinguishing_words}. It is accepted from state $(p_1',q_1)$, but is not accepted from state $p_2'$ in $\mathcal{D}'_m(a,b,c)$, nor from state $q_1$ nor $q_2$ in $\mathcal{E}_n$ of $L_n(a,c,b)$, because $c$ is the last letter in this word and $\delta_b=\delta''_c$ in $\mathcal{E}_n$ does not map any state to $n-2$.
If $p_2' \in Q'_M$ we can use the symmetric argument.

It remains to consider the case when $p_1',p_2' \in \{(m-2)',(m-1)'\}$.
Without loss of generality let assume that $p'_1 = (m-2)'$ and $p_2' = (m-1)'$.
If $q_2 \neq n-2$, then $((m-2)',q_1)$ is final but $((m-1)',q_2)$ is not.
If $q_2 = n-2$, then $q_1 \in Q_M$ since otherwise the pairs are not distinguishable. So we may use $a$, because $((m-2)',q_1) \delta^P_a = ((m-1)',n-2)$ is final, but $((m-1)',n-2)\delta^P_a = ((m-1)',n-1)$ is empty.

Finally, state $(0',0)$ is distinguished from every other considered states since it the only state from which $a^2$ is accepted.

\noindent\textbf{Symmetric difference}:

The symmetric difference is recognized by the product automaton $\mathcal{D}^P$ with $F^P = (\{(m-2)'\} \times (Q \setminus \{n-2\}) \cup ((Q'_m \setminus \{(m-2)'\}) \times \{n-2\})$.

State $((m-2)',n-1)$ is equivalent to $((m-1)',n-2)$, and also state $((m-2)',n-2)$ is equivalent to $((m-1)',n-1)$.
We show that two different reachable states $(p_1',q_1)$ and $(p_2',q_2)$ that are not one of these equivalent pairs are pairwise distinguishable.
Without loss of generality we can assume that $p_1' \neq p_2'$.

If $p_1' \in Q'_M$ or $p_2' \in Q'_M$, then we can use the same distinguishing word and argumentation as we do for the union operation, since that word  $\mathcal{E}_n$ is not accepted from $q_2$ nor from $q_1$.

So suppose that $p_1',p_2' \in \{(m-2)',(m-1)'\}$.
Without loss of generality let $p_1'=(m-2)'$ and $p_2' = (m-1)'$ (since $p_1' \neq p_2'$ as we assumed before).

If either $q_1 \neq n-2$ and $q_2 \neq n-2$ or $q_1 = n-2$ and $q_2 = n-2$, then $((m-2)',q_1)$ and $((m-1)',q_2)$ differ, since exactly one of them is final.
If $q_1=n-2$ and $q_2 \in Q_M$, then state $(p_1',q_1)=((m-2)',n-2)$ is empty but state $(p_2',q_2)=((m-1)',q_2)$ is not empty.
The last cast is $q_2=n-2$ and $q_1 \in Q_M$. Then only $((m-2)',q_1)$ can accept a non-empty word (the same as that accepted by $q_1$ in $\mathcal{D}_n$).

\noindent\textbf{Intersection}:

The intersection is recognized by the product automaton $\mathcal{D}^P$ with $F^P = \{(m-2)',n-2)\}$.

We are going to show that $(m-3)(n-3)+3$ of reachable states are pairwise distinguishable. First we have the three special equivalence classes of states: empty states (e.g.~$((m-1)',n-1)$), final state $((m-2)',n-2)$, and the initial state $(0',0)$, which is the unique state from which word $a^2$ is accepted.
The remaining states have form $(p',q)$, where $p' \in Q_M'$ and $q \in Q_M$. All of them are non-empty ($a$ is accepted) and not final. So it is enough to show that two different states from these, $(p_1',q_1)$ and $(p_2',q_2)$, have different quotients.
Without loss of generality we can assume that $p_1' \neq p_2'$.

If $q_1 \in Q_M \setminus \{1,h\}$, then we can first use a word $b^i$ such that $p_2' \delta'_{b^i} = h'$.
Then $p_1' \delta'_{b^i} \neq h'$, since $b$ acts as a cycle on $Q_M$ and $p_1' \neq p_2'$.
Moreover $q_1 \delta_{c^i} \in Q_M \setminus \{1,h\}$.
Now we apply $c$, whose action maps $h'$ to $(m-2)'$.
On the other hand we have $p_1' \delta'_{b^i} \delta'_c \in Q_M$ and $q_1 \delta_{c^i} \delta_b \in Q_M$.
So finally by applying $a$ we get that $(p_1',q_1)\delta^P_{b^i} \delta^P_c \delta^P_a = ((m-2)',n-2)$ is final, but $(p_2',q_2)\delta^P_{b^i} \delta^P_c \delta^P_a = ((m-1)',r)$ for some $r \in Q \setminus \{0\}$, which is an empty state.

If $q_1 \in \{1,h\}$ and $q_2 \in Q_M \setminus \{1,h\}$, then we can the previous argument by symmetry.
So suppose that $q_1,q_2 \in \{1,h\}$.

If $p_1' \neq h'$ then we can apply $c$, which maps $q_1$ and $q_2$ in $\mathcal{E}_n$ to states of $Q_M \setminus \{1,h\}$, since $n \ge 7$ and so $2 < h < n-3$.

We have $(p_1',q_1)\delta^P_c = (p_1' \delta'_c, q_1 \delta_b) \in Q_M' \times (Q_M \setminus \{1,h\})$ and $(p_2',q_2)\delta^P_c = (p_2' \delta'_c, q_2 \delta_b) \in (Q_M' \cup \{(m-2)'\}) \times (Q_M \setminus \{1,h\})$.
If $p_2' \delta'_c \neq (m-2)'$, then we have already shown distinguishability of this pair in the previous paragraph.
If $p_2' \delta'_c = (m-2)'$, then letter $a$ distinguishes the pair.
Finally, if $p_1' = h'$ then $p_2' \neq h'$ and we can use the symmetrical argument.

\noindent\textbf{Difference}:

The difference is recognized by the product automaton $\mathcal{D}^P$ with $F^P = (\{(m-2)'\} \times Q) \setminus \{((m-2)',n-2)\}$.

We will show that $mn-2m-3n+9$ of reachable states are pairwise distinguishable.
First we have the three special equivalence classes of states: empty states (e.g.~$((m-1)',n-1)$ or $((m-2)',n-2)$), final states (e.g.~$((m-2)',n-1)$ or $((m-2)',1)$), and initial state $(0',0)$.
Initial state is distinguished from all reachable states since it is the only one from which word $ac^2$ is accepted.
In the second group we have $m-3$ different states of the form $(p',n-1)$, where $p_1' \in Q_M'$, which are pairwise distinguished in the same way as in $\mathcal{D}_m'$, and are not empty nor final.
The last group consists of states of the form $(p',q)$, where $p' \in Q_M'$ and $q \in Q_M$.
They are not final and they are non-empty, because a word $b^i c$, for some $i$, is accepted from them.
They are distinguished from states $(p_1',n-1)$ from the second group, since word $a$ is not accepted from them, but it is accepted from $(p_1',n-1)$.
Consider two distinct states $(p_1',q_1)$ and $(p_2',q_2)$ from the third group.
We will show that they are distinguishable.

If $p_1' \neq p_2'$, then let take the distinguishing word $b^{n-2-p_1}c^2$ for $p_1'$ from Table~\ref{tab:distinguishing_words}.
It is accepted from $(p_1',q_1)$, since a word ending with $c$ cannot be accepted from $q_1$ in $\mathcal{E}_n$.
Also, as a distinguishing word, it is not accepted from $p_2'$ in $\mathcal{D}'_n$ and so not from $(p_2',q_2)$ in the product automaton.

In the opposite case $p_1'=p_2'$, so $q_1 \neq q_2$.
First suppose that $p_1'=p_2' \in Q'_M \setminus \{1',h'\}$.
Let $i$ be such that $q_1 \delta_{b^i} = h$; then $q_2 \delta_{b^i} \neq h$.
We have $(p_1',q_1)\delta^P_{c^i b a} = (p_1' \delta'_{c^i b a}, q_1 \delta_{b^i c a}) = ((m-2)',n-1))$, which is a final state.
On the other hand $(p_2',q_2) \delta^P_{c^i b a} = (p_2' \delta'_{c^i b a}, q_2 \delta_{b^i c a})=((m-2)',n-2)$, which is an empty state.
Now suppose that $p_1'=p_2' \in \{1',h'\}$.
If $q_1,q_2 \in Q_M' \setminus \{h\}$, then by applying $b$ we result in the previous case, since $m \ge 7$.
Otherwise one of $q_1$ and $q_2$, say $q_1$ without loss of generality, is equal $h$.
Then we use $ba$ and obtain $(p_1',q_1) \delta^P_{b a} = ((m-2)',h \delta_{c a}) = ((m-2)',n-1)$, which is final, and $(p_2',q_2) \delta^P_{b a} = ((m-2)',q_2 \delta_{c a}) = ((m-2)',n-2)$, which is empty.
\end{proof}

\subsection{Most complex stream}

Here we define a most complex stream for all three measures of complexity.
To meet the bound for syntactic complexity an alphabet of size at least $(n-3) + ((n-2)^{n-3}-1) + (n-3)(2^{n-3}-1) = (n-2)^{n-3} + (n-3)2^{n-3} - 1$ is required, and so a witness stream cannot have a smaller number of letters.
Our stream contains the DFAs from \cite[Definition~4]{SzWi16SyntacticComplexityOfBifixFree}, which have the transition semigroup $\Wbf(n)$.

\begin{definition}[Most complex stream, {\cite[Definition~4]{SzWi16SyntacticComplexityOfBifixFree}}]\label{def:most_complex}
For $n \ge 6$, we define the language $W_n$ which is recognized by the DFA $\mathcal{W}_n$ with $Q = \{0,\ldots,n-1\}$ and $\Sigma$ containing the following letters:
\begin{enumerate}
\item $b_i$, for $1 \le i \le n-3$, inducing the transformations $(0 \to n-1)(i \to n-2)(n-2 \to n-1)$,
\item $c_i$, for every transformation of type~(2) from Definition~\ref{def:Wbf} that is different from $(0 \to n-2)(Q_M \to n-1)(n-2 \to n-1)$,
\item $d_i$, for every transformation of type~(3) from Definition~\ref{def:Wbf} that is different from $(0 \to q)(Q_M \to n-1)(n-2 \to n-1)$ for some state $q \in Q_M$.
\end{enumerate}
\end{definition}

\begin{theorem}\label{thm:most_complex_stream}
The stream $(W_n \mid n \ge 9)$ is most complex in the class of bifix-free languages:
\begin{enumerate}
\item The quotients of $W_n$ have maximal state complexity (Proposition~\ref{pro:quotients_upper_bound});
\item $W_m$ and $W'_n$ meet the bounds for union, intersection, difference, symmetric difference, where $W'_n$ is a permutationally equivalent dialect of $W_n$;
\item $W_m$ and $W_n$ meet the bound for product;
\item $W_n$ meets the bounds for reversal and star;
\item $W_n$ meets the bound for the syntactic complexity;
\item $W_n$ meets the bounds for the number of atoms and the quotient complexities of atoms (see Theorem~\ref{thm:atoms_upper_bound}).
\end{enumerate}
Moreover, the size of its alphabet is a smallest possible.
\end{theorem}
\begin{proof}
1) As in the proof of Theorem~\ref{thm:most_complex_operations}.

2) Note that in the DFAs from Definition~\ref{def:most_complex_operations}:
\begin{itemize}
\item letter $a$ acts as some letter $d_{i_1}$ in $\mathcal{W}_n$;
\item the action of letter $b$ is induced by the word $c_{i_2} c_{i_3}$, where $c_{i_2}$ induces $(0 \to n-2)(1,\ldots,n-3)(n-2 \to n-1)$, and $c_{i_3}$ induces $(0 \to n-2)(n-2 \to n-1)$;
\item the action of letter $c$ is induced by the word $c_{i_4} b_1$, where $c_{i_4}$ induces $(0 \to n-2)(1,h)(n-3,\ldots,h+1,h-1,\ldots,2)(n-2 \to n-1)$, and $b_1$ induces $(0 \to n-1)( 1 \to n-2)(n-2 \to n-1)$.
\end{itemize}
Let $W'_n$ be a permutationally equivalent dialect of $W_n$ in which the letters $c_{i_2}$ and $c_{i_3}$ are swapped with $c_{i_4}$ and $b_1$.
Then by Theorem~\ref{thm:most_complex_operations} $W_m$ and $W'_n$ meets the bounds for Boolean operations.

3) Meeting the bound for product follows directly from Theorem~\ref{thm:product}.

4) For reversal observe that every non-empty subset $S \subseteq Q_M$ is reachable from $\{n-2\}$ by the letter $d_i$ inducing the transformation $(0 \to q)(S \to n-2)(n-2 \to n-1)$. Then $\{0\}$ is reached from $\{1\}$ by the letter $d_j$ inducing the transformation $(0 \to 1)(Q_M \to n-2)(n-2 \to n-1)$, and the empty subset from $\{0\}$ by any transformation.
Two distinct subsets $S_1,S_2 \subseteq Q_M$ are distinguished by the letter $d_k$ inducing the transformation $(0 \to q)(S \to n-2)(n-2)$, where, without loss of generality, $q$ is such that $q \in S_1$ and $q \notin S_2$.
Meeting the bound for star follows from Theorem~\ref{thm:star}.

5) Since the transition semigroup of $\mathcal{W}_n$ is $\Wbf(n)$, $W_n$ meets the bound for the syntactic complexity. Also, the size of the alphabet to meet this bound cannot be reduced (\cite{SzWi16SyntacticComplexityOfBifixFree}).

6) Since the transition semigroup of the DFA from Theorem~\ref{thm:atoms_lower_bound} is a subsemigroup of $\Wbf(n)$, we have all their transformations in the transition semigroup of $\mathcal{W}_n$; hence $W_n$ also meets the bounds for atom complexities.

\end{proof}


\section{Conclusions}

\renewcommand{\arraystretch}{1.3}
{\small\begin{table}[htb]\centering
\caption{A summary of complexity of bifix-free languages for $n \ge 6$ with the minimal sizes of the alphabet required to meet the bounds.}\label{tab:complexities}
\begin{tabular}{|l|r|r|}\hline
Measure                 & Tight upper bound          & Minimal alphabet \\ \hline\hline
Union $L_m \cup L_n$                  & $mn-(m+n)$                 & 3   \\ \hline
Symmetric difference $L_m \oplus L_n$ & $mn-(m+n)$                 & 3   \\ \hline
Intersection $L_m \cap L_n$           & $mn-3(m+n-4)$              & 2   \\ \hline
Difference $L_m \setminus L_n$        & $mn-(2m+3n-9)$             & 2   \\ \hline
Product $L_m L_n$                     & $m+n-2$                    & 1   \\ \hline
Star $L_n^*$                          & $n-1$                      & 2   \\ \hline
Reversal $L_n^R$                      & $2^{n-3}+2$                & 3   \\ \hline
\multirow{2}{*}{Syntactic complexity of $L_n$} & $(n-1)^{n-3}+(n-2)^{n-3}+$ & $(n-2)^{n-3}+$ \\
                                      & $(n-3)2^{n-3}$ & $(n-3)2^{n-3} - 1$ \\ \hline
Atom complexities $\kappa(A_S)$ & The bounds from Theorem~\ref{thm:atoms_upper_bound} & $n+1$ \\ \hline
\end{tabular}\end{table}}

We completed the previous results concerning complexity of bifix-free languages. The bounds for each considered measure are summarized in Table~\ref{tab:complexities}.
Our particular contribution is exhibition of a single ternary stream that meets all the bounds on basic operations.
Then we showed a most complex stream that meets all the upper bounds of all three complexity measures.

{\small\begin{table}[htb]\centering
\caption{The minimal sizes of the alphabet in a universal most complex stream for some of the studied subclasses of regular languages.}\label{tab:most_complex_streams}
\begin{tabular}{|l|r|}\hline
Class                                              & Minimal alphabet of a most complex stream(s) \\ \hline
Regular languages \cite{Brz13}                     & $3$ \\ \hline
Right ideals      \cite{BDL15}                     & $4$ \\ \hline
Left ideals       \cite{BDL15}                     & $5$ \\ \hline
Two-sided ideals  \cite{BDL15}                     & $6$ \\ \hline
Prefix-free \cite{BrSi16PrefixConvex}              & $n+2$ \\ \hline
Prefix-closed \cite{BrSi16PrefixConvex}            & $4$ \\ \hline
$k$-proper prefix-convex \cite{BrSi16PrefixConvex} & $7$ \\ \hline
Suffix-free \cite{BrSz15ComplexityOfSuffixFree}    & $\le 3$ and $5$ \\ \hline
Bifix-free (Theorem~\ref{thm:most_complex_stream}) & $(n-2)^{n-3}+(n-3)2^{n-3} - 1$ \\ \hline
Non-returning \cite{BrSi17NonReturning} & $n(n-1)/2$ \\ \hline
\end{tabular}\end{table}}

It is worth noting how the properties of prefix-free and suffix-free languages are shared in the class of bifix-free languages.
It is known that there does not exist such a stream in the class of suffix-free languages, even considering only basic operations.
Hence, although the classes of bifix-free and suffix-free languages share many properties, such as a similar structure of the largest semigroups, the existence of most complex languages distinguishes them.
This is because the bounds for star and product are much smaller for bifix-free languages and are very easily met.
Additionally, a most complex stream of bifix-free languages requires a superexponential alphabet, which is much larger than in most complex streams of the other studied subclasses; see Table~\ref{tab:most_complex_streams}.

\bibliographystyle{plain}

\end{document}